\documentclass[a4paper,11pt]{article}

\usepackage{fullpage}
\usepackage{pdfpages}
\usepackage[T1]{fontenc}
\usepackage{microtype}
\usepackage{mathtools,amsmath,xspace}
\usepackage{amssymb}
\graphicspath{{figures/}}
\title{Finding Pairwise Intersections of Rectangles in a Query
  Rectangle%
  \footnote{This research was Supported by the MSIT(Ministry of
    Science and ICT), Korea, under the SW Starlab support
    program(IITP--2017--0--00905) supervised by the IITP(Institute for
    Information \& communications Technology Promotion and the NRF
    grant 2011-0030044 (SRC-GAIA) funded by the government of Korea.}}

\author{Eunjin Oh\thanks{Department of Computer Sience and
    Engineering, POSTECH, Korea, \texttt{\{jin9082,
      heekap\}@postech.ac.kr}} \and Hee-Kap Ahn\footnotemark[2]}

\newtheorem{theorem}{Theorem} 
\newtheorem{lemma}[theorem]{Lemma}

\newtheorem{corollary}[theorem]{Corollary}
\newtheorem{definition}[theorem]{Definition}
\newtheorem{observation}[theorem]{Observation}

\newcommand{\sides}{\ensuremath{\mathcal{S}_\ell}}
\newcommand{\itx}[2]{\ensuremath{I(#1,#2)}}
\newcommand{\ssi}[1]{\ensuremath{\mathcal{S}_C(#1)}}
\newcommand{\ssb}[1]{\ensuremath{\mathcal{S}_B(#1)}}
\newcommand{\sss}[1]{\ensuremath{\mathcal{S}(#1)}}

\newcommand{\slab}[1]{\ensuremath{H(#1)}}

\newcommand{\proj}[2]{\ensuremath{\mathsf{#1}(#2)}}

\newcommand{\axis}[1]{\ensuremath{\mathsf{x}_{#1}}}
\newcommand{\iset}{\ensuremath{\mathcal{S}}}
\newcommand{\oset}{\ensuremath{\mathcal{U}}}
\newcommand{\sset}{\ensuremath{\mathcal{S}}}
\newcommand{\confa}{\textsf{C1}\xspace}
\newcommand{\confb}{\textsf{C2}\xspace}
\newcommand{\confc}{\textsf{C3}\xspace}
\newcommand{\confd}{\textsf{C4}\xspace}
\newcommand{\confe}{\textsf{C5}\xspace}

\newcommand{\SegInt}{\textsc{SegInt}\xspace}
\newcommand{\EPtEnc}{\textsc{EPtEnc}\xspace}
\newcommand{\PtEnc}{\textsc{PtEnc}\xspace}
\newcommand{\RecEnc}{\textsc{RecEnc}\xspace}
\newcommand{\RecCross}{\textsc{RecCross}\xspace}
\newcommand{\RecInt}{\textsc{RecInt}\xspace}

\newcommand{\GridCont}{\textsc{GridCont}\xspace}
\newcommand{\BoxInt}{\textsc{BoxInt}\xspace}
\newcommand{\PairFind}{\textsc{PairFind}\xspace}
\newcommand{\BoxPairInt}[1]{\textsc{BoxPairInt}[#1]\xspace}

\definecolor{magenta(process)}{rgb}{1.0, 0.0, 0.56}

\newbox\ProofSym \setbox\ProofSym=\hbox{%
	\unitlength=0.18ex%
	\begin{picture}(10,10) \put(0,0){\framebox(9,9){}}
	\put(0,3){\framebox(6,6){}}
	\end{picture}}

\def\polylog{\operatorname{polylog}}

\begin{document}
\date{}
\maketitle

\begin{abstract}
  We consider the following problem: Preprocess a set $\iset$ of $n$
  axis-parallel boxes in $\mathbb{R}^d$ so that given a query of an
  axis-parallel box in $\mathbb{R}^d$, the pairs of boxes of
  $\iset$ whose intersection intersects the query box can be reported
  efficiently.  For the case that $d=2$, we present a data structure
  of size $O(n\log n)$ supporting $O(\log n+k)$ query time, where $k$
  is the size of the output.  This improves the previously best known
  result by de Berg et al. which requires $O(\log n+ k\log n)$
  query time using $O(n\log n)$ space.  There has been no result known
  for this problem for higher dimensions, except that for $d=3$, the
  best known data structure supports $O(\sqrt{n}\log^2n+k\log^2n)$
  query time using $O(n\sqrt {n}\log n)$ space.  For a constant $d>2$,
  we present a data structure supporting
  $O(n^{1-\delta}\log^{d-1} n + k \polylog n)$ query time for any
  constant $1/d\leq\delta<1$.  The size of the data structure is
  $O(n^{\delta d - 2\delta + 1}\log n)$.
\end{abstract}

\section{Introduction}
Range searching is one of the fundamental problems, which has been studied
extensively in computational
geometry~\cite{agarwal1999geometric}. Typical problems of this type
are formulated as follows.
\begin{quote}
  Preprocess a set $\mathcal{I}$ of input
  geometric objects so that given a query of geometric object $Q$, the
  objects in $\mathcal{I}$ intersecting $Q$ can be reported or counted efficiently.
\end{quote}
There are a number of variants of the problem including 
checking if an object in $\mathcal{I}$
intersects $Q$, finding the minimum (or maximum) weight of the
objects in $\mathcal{I}$ intersecting $Q$, and computing the sum of the weights of
the objects in $\mathcal{I}$ intersecting $Q$.

In this paper, we consider a variant of the range searching problem, which
is stated as follows. Given a set $\iset$ of $n$ axis-parallel boxes in $\mathbb{R}^d$,
preprocess $\iset$ so that given a query of an axis-parallel box $Q$
in $\mathbb{R}^d$, all the pairs $(S,S')$ of boxes of $\iset$ with
$S\cap S'\cap Q\neq \emptyset$ can be reported efficiently. The 
desired running time for the query algorithm is of form
$O(f(n)+k(g(n)))$ for some functions $f(n)=o(n)$ and $g(n)=o(n)$, where
$k$ is the size of the output.  One straightforward way is to compute
all boxes of $\iset$ intersecting $Q$ and to check whether each pair
$(S,S')$ of them has their intersection point in $Q$.  However,
this straightforward algorithm takes $\Omega(n)$ time  in the worst case
even when $k=0$.

This problem occurs in a number of real-world applications. For instance,
suppose that we are given a collection of personal qualities (or personality
traits) of $n$ clients stored in a database, each of them is represented as
an interval of values. A pair of clients is said to be compatible each other
if there is a common subinterval over every quality of them. A typical query
on such a collection is composed of a range on each of the qualities,
which represents a certain criterion of selecting some compatible pairs of
clients that match the query criterion. 

If we are allowed to use $\Omega(n^2)$ space in the database, we may
precompute all compatible pairs in advance and store them to answer queries efficiently.
Otherwise, it is desirable to devise a way of storing the data using
less amount of space while the query time remains the same or does not increase
much. That is, we need to construct a data structure to answer such a query
efficiently in both the query time and the size of the data structure. This is the goal of
the problem we study in this paper.

\paragraph{Previous Work.} There are a few results on this
problem~\cite{deBerg2015,Gupta,Rahul}.  Consider a simpler problem in
which input objects are orthogonal line segments. Orthogonal line
segments can be considered as degenerate axis-parallel
rectangles. Gupta~\cite{Gupta} presented a data structure of size
$O(n\log^2 n)$ supporting $O(\log^2 n+k)$ query time for this problem,
where $k$ is the size of the output and $n$ is the size of the
input. Later, the size of the data structure and the query time were
improved to $O(n\log n)$ and $O(\log n+k)$, respectively by Rahul et
al.~\cite{Rahul}.

For axis-parallel rectangles in the plane, de Berg et
al.~\cite{deBerg2015} presented a data structure of size $O(n\log n)$
that supports $O(\log n\log^* n+k\log n)$ query time. We observe that
their data structure can be improved to support $O(\log n+ k\log n)$
query time by simply replacing the range searching algorithm
in~\cite{Ptree} with the one in~\cite{Afs}. For details, see
Section~\ref{sec:ds-easycase}.  In fact, this is mentioned in the
journal paper~\cite{deBerg2015-journal} by the authors, which has been
available online recently with query time $O(\log n+ k\log n)$.
The algorithm by de Berg et
al.~\cite{deBerg2015-journal,deBerg2015} does not extend to higher dimensions directly.
Using more observations and techniques, they presented a data
structure of size $O(n\sqrt{n}\log n)$ supporting
$O(\sqrt{n}+k\log^2 n\log^* n)$ query time in
$\mathbb{R}^3$~\cite{deBerg2015}.~\footnote{The
journal paper presents $O(\sqrt{n} \log^2 n+ k\log^2 n)$
query time with the same space complexity~\cite{deBerg2015-journal}.}

One might be concerned on the preprocessing time as well as the size
of the data structure. In this type of problems, however, queries are
supposed to be made in a repetitive fashion and the preprocessing time
can be seen as being amortized over the queries to be made later
on~\cite{chazelle1986Enclusure}. Therefore, we focus mainly on the
space requirement of the data structure and the query time for the
problem as other previous works did.

\paragraph{Our Result.} In this paper, we first present a data structure 
of size $O(n\log n)$ for two-dimensional case that supports $O(\log n+k)$ query time.
This improves the data structure of de Berg et al.~\cite{deBerg2015-journal}. 
Recall that our problem is a generalization of the problem studied by Rahul et al.~\cite{Rahul}.
Although our problem is more general, our data structure with its query algorithm
requires the same storage and running time as theirs.

Moreover, our data structure is almost optimal. 
To see this, observe that our problem can be reduced to the 
\emph{2D orthogonal range reporting problem}. Given a set $\mathcal{P}$ of points in $\mathbb{R}^2$,
the 2D orthogonal range reporting problem asks to preprocess
them so that given a query of an axis-parallel rectangle,
the points of $\mathcal{P}$ contained in the query rectangle can be reported efficiently.
To solve this problem using a data structure for our problem,
we map each point $p$ in $\mathcal{P}$
to two points lying on $p$ (two degenerate boxes).
Then we construct a data structure for our problem on the set of the degenerate boxes for all points in $\mathcal{P}$.
The data structure reports the pairs $(S,S')$ of degenerate boxes
such that $S$ and $S'$ lie on the same position and are contained in a query rectangle.
Therefore, we can answer the 2D orthogonal range reporting problem using the
data structure for our problem
without increasing the running time.
For the 2D orthogonal range reporting problem, it is known that on a pointer machine model, a query time of $O(\polylog n+k)$, where $k$ is the size of the output, can only be achieved at the expense of $\Omega(n\log n/\log\log n)$ storage~\cite{chazelle1990lower}.
Moreover, on a pointer machine model, a query time of $o(\log n+k)$ cannot be achieved
	regardless of the size of the data structure.
Therefore, our query time is optimal, and the size of our data structure is almost optimal.

We also consider the problem in higher dimensions $\mathbb{R}^d$.
For a constant $d>2$, we present a data structure 
that supports $O(n^{1-\delta}\log^{d-1} n + k \log^{d-1} n)$ query time for any constant $\delta$ with $1/d\leq\delta<1$.
The size of the data structure is $O(n^{\delta d - 2\delta + 1}\log n)$.
A constant $\delta$ shows a trade-off between storage and query time.
This is the first result on the problem in higher dimensions.
\medskip

Throughout the paper, we use $\iset=\{S_1,\ldots, S_n\}$ to denote a given set of
$n$ axis-parallel boxes in $\mathbb{R}^d$ for a
constant $d\geq 2$.
For any two boxes
$S_i, S_j\in \iset$,  we use $\itx{i}{j}$ to denote the
intersection of $S_i$ and $S_j$.
Our goal is to preprocess $\iset$ so that for a query
of an axis-parallel box $Q$, we can report all pairs $(S_i,S_j)$ of
boxes of $\iset$ with $I(i,j)\cap Q\neq \emptyset$
efficiently.  We use $\oset(Q)$ and $k(Q)$ to denote the output
and the size of the output for a query $Q$, respectively.  We simply
use $\oset$ and $k$ to denote
$\oset(Q)$ and $k(Q)$, respectively, if they are understood in context.

\section{Planar Case}\label{sec:planar}
In this section, we consider the problem in the plane, that is, we are
given a set $\iset$ of $n$ axis-parallel rectangles in the plane.
We present a data structure of size $O(n\log n)$ that supports $O(\log n+k)$ query
time for queries of axis-parallel rectangles.  This improves the
previously best known data structure with its query algorithm by de
Berg et al.~\cite{deBerg2015-journal}.
Their data structure has size $O(n\log n)$
and supports $O(\log n+k\log n)$ query time~\cite{deBerg2015-journal}.

\subsection{Configurations of Two Intersecting Rectangles}
An axis-parallel rectangle has four sides: the top, bottom, left and
right sides.  We call the top and bottom sides the \emph{horizontal
  sides}, and the left and right sides the \emph{vertical sides}.

Consider a side $ab$ of a rectangle $S\in\sset$ with endpoints $a$ and $b$.
Let $a'b'$ be the segment on $ab$
such that $a'$ and $b'$ are the points closest to $a$ and $b$,
respectively, among all intersection points of $ab$ with input rectangles
other than $S$. We call $a'b'$ the \emph{stretch} of $S$ on $ab$. Note
that $ab$ has no stretch if $ab$ intersects no
rectangles of $\sset\setminus\{S\}$. The stretch of $ab$ is $ab$ if
$a$ and $b$ are contained in some rectangles of $\sset$ other than $S$. 
There is at most one stretch for each side of a rectangle of
$\iset$.  Let $\sides$ be the set of all stretches of the rectangles of $\iset$.

For any pair $(S_i, S_j)$ of rectangles of $\iset$ with
  $\itx{i}{j}\cap Q\neq\emptyset$, it is not difficult to see that the
  pair belongs to one of the following three cases: (1) $Q$ is
  contained in one of the two rectangles of the pair, (2) $Q$ contains
  a corner of $\itx{i}{j}$, or (3) $Q$ intersects the boundary of $\itx{i}{j}$, but
    contain no corner of $\itx{i}{j}$.
  Here we propose another way of describing all the
  cases in terms of stretches so that the query time can be improved
  without increasing the size of the data structures compared to the one
  in~\cite{deBerg2015}.  
  Each of these
  cases can be rephrased into one or two configurations in
  Observation~\ref{obs:config}. More precisely, case (1) corresponds
  to \confa, case (2) corresponds to \confb~and \confc, and case (3)
  corresponds to \confd~and \confe~of Observation~\ref{obs:config}.
\begin{figure}
  \begin{center}
    \includegraphics[width=0.8\textwidth]{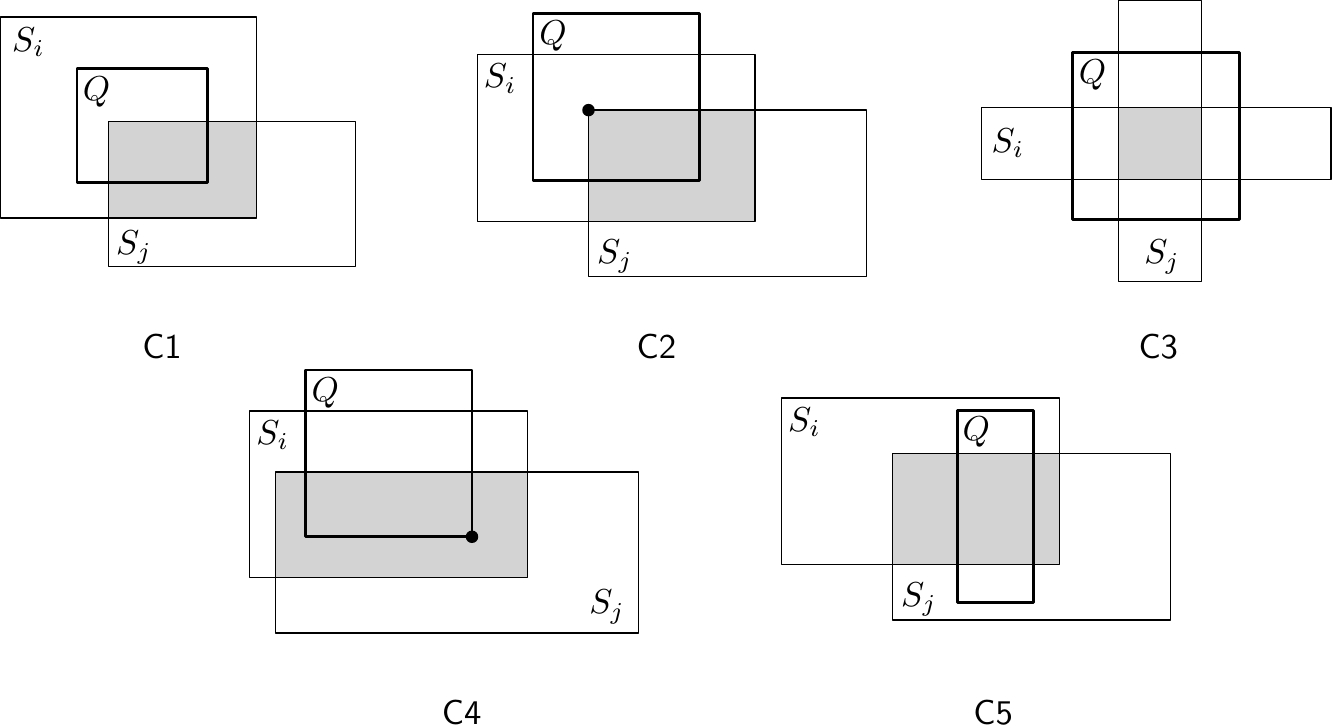}
    \caption{\label{fig:config} Five configurations of $(S_i,S_j)$ and
      $Q$.}
  \end{center}
\end{figure}
\begin{observation}[\textsf{Five Configurations of Intersections.}]\label{obs:config}
  For any pair $(S_i,S_j)$ of rectangles of $\iset$ with
  $\itx{i}{j}\cap Q\neq\emptyset$, one of the followings
  holds. Figure~\ref{fig:config} gives an illustration.
  \begin{itemize}
  \item \confa. $S_i$ or $S_j$ contains $Q$.
  \item \confb. $Q$ contains 
      an endpoint of a stretch of $S_i$ or $S_j$ which is a corner of $\itx{i}{j}$. 
  \item \confc. A stretch of $S_i$ and a stretch of $S_j$ cross $Q$ in
    different directions.
    \item \confd. $\itx{i}{j}$ contains a corner of $Q$.
    \item \confe. 
      $\itx{i}{j}$ and $Q$ cross each other.
  \end{itemize}
\end{observation}

We consider the configurations one by one in our query algorithm. We
first report all pairs satisfying \confa (simply, all
\emph{\confa-pairs}), then we report all pairs satisfying \confb
(simply, all \emph{\confb-pairs}), and so on.  There might be a pair
$(S_i,S_j)$ of input rectangles that belongs to more than one configuration. 
To avoid reporting the same pair more than once, we give a priority order
to the configurations such that our algorithm reports a pair exactly once
in the configuration of the highest priority
among the configurations the pair belongs to.
Since there are only five
configurations and we can check in constant time whether a pair belongs to
a configuration or not, this does not increase the asymptotic time complexity
of our algorithm.

\subsection{Reporting All Pairs, except \confe-pairs}
We first show how to construct data structures for finding all
pairs $(S_i,S_j)$ of input rectangles with $\itx{i}{j}\cap Q\neq\emptyset$, 
except \confe-pairs. In Section~\ref{sec:D1}, we show how to find all
\confe-pairs.

\subsubsection{Data Structures}\label{sec:ds-easycase}
We construct four data structures for four different problems: the
\emph{orthogonal segment intersection problem}, the \emph{point
  enclosure problem}, the \emph{orthogonal range reporting
  problem}, and the \emph{rectangle crossing problem}.  There has been
a fair amount of work on these problems. 
We observe that the last problem reduces to the 3D orthogonal range
reporting problem with a four-sided query box, which has also been
studied well.  Thus we use data structures for these four
problems after slightly modifying them to achieve our purpose.

\paragraph{Orthogonal Segment Intersection Problem: {\mdseries\SegInt}.}
The \emph{orthogonal segment intersection problem} asks to preprocess
horizontal input  segments so that given a query of a vertical segment,
the horizontal input  segments intersected by the query can be computed efficiently.
Chazelle~\cite{chazelle1986Enclusure} gave a data structure called the
\emph{hive-graph} to solve this problem efficiently.  The
hive-graph is a planar orthogonal graph with $O(N)$ cells, each of
which has a constant number of edges on its boundary, where $N$ is
the number of the input segments.

The query algorithm first finds the cell of the
hive-graph containing an endpoint of the query segment and traverses
the hive-graph along the query segment from the endpoint to the other
endpoint. All horizontal edges intersected by the query are encountered during the traversal.
In this way, the algorithm finds all horizontal segments intersected by the query in order sorted along the
query.  The query algorithm takes constant time per
output segment, excluding the time for the point location for an endpoint of the query.

In our problem, we construct two hive-graph data structures, one for
the horizontal sides of the rectangles of $\iset$ and one for the vertical
sides of the rectangles of $\iset$.
The query segments used in our query algorithm are stretches of $\sides$. 
To save the time for point locations in the query algorithm, for each endpoint of the stretches of $\sides$,
we find the two cells of
the two hive-graphs that contain the endpoint in the preprocessing phase.
Due to this preprocessing, we can find the sides of the rectangles of $\iset$ crossed by a stretch
$\ell$ of $\sides$ in the sorted order along $\ell$ from one
endpoint of $\ell$ in constant time per output side.
We denote this data structure by \SegInt. 

\paragraph{Point Enclosure Problem: {\mdseries\PtEnc} and {\mdseries\EPtEnc}.}  The \emph{point enclosure
  problem} asks to preprocess input rectangles so that all input
rectangles containing a query point can be computed efficiently.
Chazelle~\cite{chazelle1986Enclusure} gave a data structure for this problem.
We construct this data structure on $\iset$ in the preprocessing time,
and denote the data structure by \PtEnc.
It has size $O(n)$ and allows us to find all rectangles of $\iset$ containing
a query point  in $O(\log n+K)$ time, where $K$ is the size of the output
in this subproblem. 
Moreover, it allows us to check whether there exists such a rectangle in $O(\log n)$ time.

In our query algorithm, we consider this problem for two different
purposes: finding all rectangles of $\iset$ containing a
\emph{corner} of $Q$, and finding all rectangles of $\iset$
containing an \emph{endpoint} of a stretch of $\sides$.  
We perform the former task at most four times in our query algorithm
since $Q$ has four corners. Thus we simply use \PtEnc for this task. 
However, we will perform the latter task
$\Theta(k)$ times in the worst case, which takes $\Omega(k\log n)$ time.
Here $k$ is the size of the output in our query algorithm.
Note that we
have the endpoints of the stretches of $\sides$ in the preprocessing
phase, and therefore the latter task can be done 
in the preprocessing phase.

To do this, we show how the data structure by
Chazelle~\cite{chazelle1986Enclusure} works.  Its primary structure is
a balanced binary search tree on the rectangles of $\iset$ with respect
to the $x$-coordinates of their vertical sides. Each node of the
binary search tree corresponds to a vertical line, and it is augmented
by the hive-graph on the set of the rectangles of $\iset$
intersecting its corresponding vertical line. 
The query algorithm finds $O(\log n)$ nodes of the binary search tree,
and then searches on the hive-graphs associated with the nodes. This
takes $O(\log n+K)$ time due to fractional cascading, where $K$ is the
size of the output in this subproblem.

This means that we consider $O(\log n)$ hive-graphs and spend
$O(\log n)$ time to find the cell containing a query point on one
hive-graph. The point location on the other hive-graphs can be done by
fractional cascading. To save the $\log n$ term in the running time of
the query algorithm, we find the cells of the $O(\log n)$ hive-graphs
containing each endpoint of the stretches of $\sides$ in the
preprocessing time. 
We need $O(n\log n)$ space to store the cells containing
endpoints of the stretches of $\sides$.  Due to the preprocessing,
given an endpoint of a stretch of $\sides$,
we can find all rectangles of $\iset$ containing the endpoint in $O(1+K)$ time.  Note
that $O(1+K)=O(K)$ since each endpoint is contained in at least two
rectangles of $\iset$, and thus $K>1$.
We denote this data structure (\PtEnc associated with pointers
for the endpoints of the stretches) by \EPtEnc.

\paragraph{Orthogonal Range Reporting Problem: {\mdseries\RecEnc}.}  We
want to preprocess all endpoints of the stretches of $\sides$ so that
the endpoints contained in a query rectangle can be computed
efficiently. Chazelle~\cite{chazelle1986Enclusure} presented a
data structure for this problem that has
$O(n\log n/\log\log n)$ size and supports $O(\log n+K)$ query time, where
$K$ is the size of the output.
We denote this data structure by \RecEnc.

\paragraph{Rectangle Crossing Problem: {\mdseries\RecCross} and {\mdseries\RecInt}.}  
We want to preprocess the
stretches of $\sides$ so that all stretches 
crossing a query rectangle can be computed efficiently.  De Berg et
al.~\cite{deBerg2015} also considered this problem. To do this, they
reduce this problem to the orthogonal range reporting problem in three
dimensional space as follows.  Let $[a,b]\times[c,d]$ be a query
rectangle. The query rectangle is crossed by a vertical stretch
$x_1\times [y_1,y_2]$ if and only if $x_1\in [a,b]$,
$y_1\in[-\infty,c]$, and $y_2\in[d,\infty]$.  Using this observation,
they map each vertical stretch $x_1\times [y_1,y_2]$ to the point
$(x_1,y_1,y_2)$ in $\mathbb{R}^3$. Then we can find all vertical stretches
crossing the query rectangle by finding all points contained in the
orthogonal region $[a,b]\times[-\infty,c]\times[d,\infty]$.
Similarly, we can do this for horizontal stretches.  However, they did not
use the fact that a query is unbounded: it  is four-sided in $\mathbb{R}^3$.
In this case, we can use a more efficient algorithm given
by Afshani et al.~\cite{Afs} instead of the one in~\cite{Ptree}.
In fact, this is also mentioned in the journal paper~\cite{deBerg2015-journal}
by the authors, which has been available online recently.
The algorithm by Afshani et al. takes $O(\log n+K)$ time for four-sided query boxes
using a data structure of $O(n\log n/\log\log n)$ size, where $K$ is the size of
the output.
We denote this data structure by \RecCross.
This data structure has size $O(n\log n/\log\log n)$
and allows us to find all  vertical (or horizontal) 
stretches of $\sides$ crossing a query rectangle in $O(\log n+K)$ time,
where $K$ is the size of the output. 

A rectangle $S$ of $\iset$ intersects a query rectangle $Q$ if and only if
(1) $Q$ crosses a side of $S$, (2) $Q$ contains a corner of $S$,
or (3) $Q$ is contained in $S$.
To find all rectangles of $\iset$ intersecting a query rectangle,
we use \RecCross for case (1), use \RecEnc for case (2), and use
\PtEnc for case (3). 
We call the combination of these data structures \RecInt. 
We can find all rectangles of $\iset$ intersecting $Q$ in $O(\log n+K)$ time
using \RecInt, where $K$ is the size of the output in this subproblem.

\subsubsection{Query Algorithms.}
Assume that we have the data structures of size $O(n\log n)$ described
in Section~\ref{sec:ds-easycase}.  Then, we can find all pairs
$(S_i,S_j)$ of $\iset$ with $\itx{i}{j}\cap Q\neq\emptyset$, except \confe-pairs, in $O(\log n+k)$
time.

\paragraph{Reporting \confa-pairs of $Q$.}
We can find the \confa-pairs of $Q$ in $O(\log n+k(Q))$ time. A pair
of rectangles of $\iset$ is a \confa-pair of $Q$ if one rectangle of the
pair contains all four corners of $Q$ and the other rectangle
intersects $Q$.

We find the rectangles of $\iset$ containing all four corners of
$Q$ by finding all rectangles of $\iset$ containing each corner of $Q$
using \PtEnc. Note that there are
$O(k(Q)+1)$ rectangles that contain a corner of $Q$ simply because
every pair of the rectangles containing the corner is in $\oset(Q)$.
(We need ``$+1$'' since it is possible that there is just one rectangle containing the corner,
but $k(Q)$ is zero.)
Thus, we can compute such rectangles in $O(\log n+k(Q))$ time.
Let $\iset_1$ denote
the set of all rectangles containing all four corners of $Q$.

If $\iset_1$ is not empty, we find all rectangles intersecting $Q$ 
in $O(\log n+K)$ time using \RecInt, where $K$ is the number of such rectangles.
Since $\iset_1$ is not empty, $K$ is at most $k(Q)$.
Let $\mathcal{S}_2$ be the set of all rectangles intersecting $Q$. 
We report every pair $(S_1,S_2)$ with $S_1\in \iset_1$ and $S_2\in\iset_2$ as a \confa-pair of $Q$,
which takes $O(\log n+k(Q))$ time.
It is clear that we report all \confa-pairs of $Q$ in this way.

\paragraph{Reporting \confb-pairs of $Q$.}  We can find the
\confb-pairs of $Q$ in $O(\log n+k(Q))$ time. A pair of
rectangles of $\iset$ is a \confb-pair of $Q$ if $Q$ contains
an endpoint of a stretch $\ell$ of one of them and the other intersects $\ell\cap Q$.
We find all stretches of $\sides$ whose
endpoints are in $Q$ in $O(\log n+k(Q))$ time using \RecEnc. The number of such stretches is $O(k(Q))$
because each endpoint of the stretches of $\sides$ is contained in at
least two rectangles of $\iset$ and there are at most four stretches from one rectangle of $\iset$.

For each stretch $\ell$ with an endpoint in $Q$, we want to find all rectangles $S$ of
$\iset$ with $S\cap \ell \cap Q\neq\emptyset$.
Such rectangles
$S$ satisfy one of the followings: $\ell\cap Q$ is intersected by the
boundary of $S$ or $\ell\cap Q$ is contained in $S$.  For the former
case, we use \SegInt.  Starting from the endpoint of $\ell$
contained in $Q$, we traverse the hive-graph along $\ell$ until we
escape from $Q$ or we arrive at the other endpoint of $\ell$.  We find all rectangles $S$ whose sides intersect
$\ell\cap Q$ in time linear in the number of such rectangles using \SegInt.
For the latter case, we compute all rectangles
containing the endpoint of $\ell$ that is also in $Q$ in time linear in the
number of such rectangles using $\EPtEnc$.
Therefore, for each stretch $\ell$ with an endpoint in $Q$, we can find
all rectangles of $\iset$ intersecting $\ell\cap Q$ in time
linear in the number of such rectangles.

By applying this procedure for every stretch with an endpoint in $Q$,
we can find all \confb-pairs of $Q$ in $O(k(Q))$ time, excluding the
time for finding all such stretches.  Therefore, we can compute all
\confb-pairs of $Q$ in $O(\log n+k(Q))$ time in total.

\paragraph{Reporting \confc-pairs of $Q$.}  We can find the
\confc-pairs of $Q$ in $O(\log n+k(Q))$ time.  A pair of
rectangles of $\iset$ is a \confc-pair of $Q$ if two stretches, one from
each rectangle, cross $Q$ in different directions.  Let $\iset_v$
be the set of the rectangles of $\iset$ whose vertical stretches cross
$Q$.  Let $\iset_h$ be the set of the rectangles of $\iset$ whose
horizontal stretches cross $Q$.

We first check whether $\iset_v$ or $\iset_h$ is empty in $O(\log n)$
time using \RecCross. If one of them is empty, there is no \confc-pair of $Q$. 
If both of them are
nonempty, we compute $\iset_v$ and $\iset_h$ in $O(\log n+k(Q))$ time using \RecCross.
The size of $\iset_v$ and
$\iset_h$ is $O(k(Q))$ since every rectangle of $\iset_v$ intersects 
every rectangle of $\iset_h$ in $Q$.  Then we report the pairs $(S,S')$ with
$S\in\iset_v$ and $S'\in\iset_h$ as the \confc-pairs in
$O(\log n+k(Q))$ time in total.

\paragraph{Reporting \confd-pairs of $Q$.}  We can report the
\confd-pairs of $Q$ in $O(\log n+k(Q))$ time.  A pair of
rectangles of $\iset$ is a \confd-pair of $Q$ if the intersection of the rectangles
contains a corner of $Q$. In this case, both rectangles of the pair contains a corner of $Q$.
We first check whether
there exists a rectangle of $\iset$ containing a corner of $Q$
in $O(\log n)$ time using \PtEnc.
Again, the number of the rectangles of $\iset$ containing a corner of $Q$ is
$O(k(Q))$ as every pair of such rectangles is in $\oset(Q)$.
If there exists a rectangle containing a corner of $Q$, we find all rectangles containing
the corner of $Q$ in
$O(\log n+k(Q))$ time using \PtEnc.  Then we report
all pairs consisting of the rectangles containing the corner of $Q$.  We do this for each of the other
corners of $Q$. Then we can report all \confd-pairs in
$O(\log n+k(Q))$ time.

\subsection{Reporting \confe-pairs}\label{sec:D1}
We have shown how to find all pairs of rectangles of
$\iset$ intersecting each other in $Q$, except for the
\confe-pairs. 
There might be some pairs of rectangles that belong to both
\confe~and one of the other configurations. As mentioned earlier,
this can be checked in constant time per pair of rectangles. Since we use a
priority order over the configurations, we assume that they have
already been reported by the algorithm for the configurations other
than \confe.

A pair of rectangles of $\iset$ is a \confe-pair of a query rectangle $Q$ if the intersection of the rectangles 
and $Q$ cross each other.   
In the following, we show how to find and report the \confe-pairs of $Q$
not belonging to any other configuration such that the horizontal sides of the intersection 
intersect the vertical sides of $Q$.
The \confe-pairs not belonging to any other configuration
such that the vertical sides of the intersection 
intersect the horizontal sides of $Q$ can be found analogously.

\paragraph{One-Dimensional Segment Tree.}
We construct a \emph{one-dimensional segment tree} $T$ of
$\iset$ with respect to the $x$-axis as follows~\cite{CGbook}.
The segment tree is a balanced binary search tree on the orthogonal projections of 
the rectangles of $\iset$ onto the $x$-axis. 
Each node $v$ of
the balanced binary search tree corresponds to a closed vertical slab
$\slab{v}$. The union of all vertical slabs corresponding
to the nodes at the same level is $\mathbb{R}^2$. 
We say that a rectangle $S$ \emph{crosses} $\slab{v}$ if $S$ intersects
$\slab{v}$ and  no vertical side of $S$ is contained in $\slab{v}$.  Let
$\ssi{v}$ be the set of the rectangles of $\iset$ that
cross $\slab{v}$ but do not cross $\slab{u}$ for the parent
$u$ of $v$ in $T$.  There are $O(\log n)$ nodes $v$ with
$S\in\ssi{v}$ for a rectangle $S$. Moreover, the union of $\slab{v}$'s for all such nodes
$v$ contains $S$. 
Let $\ssb{v}$ be the set of the rectangles of
$\iset$ whose left or right side is contained in the interior of $\slab{v}$.
Note that $\ssb{v}$ is empty for every leaf node $v$. 
For a rectangle $S\in\iset$, there are at most two nodes $v$
of $T$ with $S\in\ssb{v}$ at each level of $T$, and
each such node lies on one of the two paths of $T$ from the root to
two leaf nodes $w, w'$ with the left side of $S$ contained in $\slab{w}$ and
the right side of $S$ contained in $\slab{w'}$. 
We use $\sss{v}$ to denote the union of $\ssi{v}$
and $\ssb{v}$.  For each node $v$ of $T$, we store $\ssb{v}$ and
$\ssi{v}$. The binary search tree together with the sets
$\ssb{\cdot}$ and $\ssi{\cdot}$ forms the segment tree of
$\iset$.
The size of $T$ is $O(n\log n)$.

\paragraph{Canonical Nodes of a \confe-pair.}
Consider any \confe-pair $(S_i,S_j)$ of $Q$.  There are $O(\log n)$
nodes $v$ of $T$ such that $S_i, S_j\in\sss{v}$ 
and $\itx{i}{j}\cap Q\cap H(v)\neq\emptyset$.
This means that there can be $\Omega(k\log n)$ such nodes in the worst case
for all \confe-pairs in total.
Instead of considering all of them, we 
use \emph{canonical nodes} (to be defined below) such that
there is a unique canonical node of $(i,j,Q)$ in $T$ for any $\confe$-pair.
We will show how to find the canonical nodes and report all \confe-pairs efficiently
in the subsequent sections.
See Figure~\ref{fig:node}.

\begin{definition}\label{def:ijq}
  For a rectangle $Q$ and a pair $(S_i,S_j)$ of the rectangles of
  $\iset$ with $\itx{i}{j}\cap Q\neq\emptyset$, a node $v$
  of $T$ is called the \emph{canonical node} of $(i,j,Q)$ if the left
  side of $Q$ is contained in $\slab{v}$ and
  both $S_i$ and $S_j$ are in $\sss{v}$ satisfying
  $S_i\in\ssi{v}$ or $S_j\in\ssi{v}$.
\end{definition}

Note that not every canonical node of some triple $(i,j,Q)$ defines
a \confe-pair of $Q$, though $\itx{i}{j}\cap Q\neq\emptyset$.
However, there is a canonical node of $(i,j,Q)$ in $T$ for each
\confe-pair of $Q$ such that the horizontal sides of 
	$\itx{i}{j}$ intersect the vertical sides of $Q$.
\begin{figure}
  \begin{center}
    \includegraphics[width=0.75\textwidth]{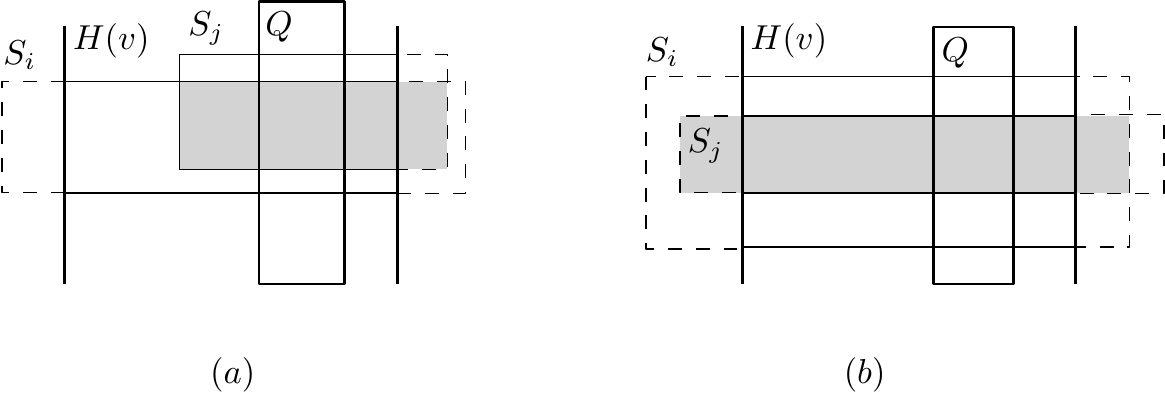}
    \caption{\label{fig:node}(a) A canonical node $v$ of $(i,j,Q)$. It
      holds that $S_i\in\ssi{v}$ and $S_j\in\ssb{v}$. (b) A canonical
      node $v$ of $(i,j,Q)$. It holds that both $S_i$ and $S_j$ are in
      $\ssi{v}$.}
  \end{center}
\end{figure}

\begin{lemma}\label{lem:exist}
  For any \confe-pair $(S_i,S_j)$ of $Q$ such that the horizontal sides of 
  $\itx{i}{j}$ intersect the vertical sides of $Q$, there is a canonical node of
  $(i,j,Q)$ in $T$.
\end{lemma}
\begin{proof}
  Consider the \confe-pairs of $Q$ such that the horizontal sides of
  the intersection intersects the vertical sides of $Q$. 
  Let $p$ be the intersection
  between the left side of $Q$ and the top side of $\itx{i}{j}$.  Then
  there is a path $\pi$ from the root node to some leaf node $u$ with
  $p\in\slab{u}$ in $T$.  Consider a node $w$ in $\pi$.  Since $p$
  lies on the left side of $Q$, the slab $\slab{w}$ contains the left
  side of $Q$.  Moreover, $\slab{w}$ intersects both $S_i$ and $S_j$.
	
  We claim that there is a canonical node of $(i,j,Q)$ in $\pi$.  By
  the construction of the segment tree, $S_i\in\ssb{v}$ for the root
  node $v$ and $S_i\not\in\ssb{u}$ for the leaf node $u$ of $\pi$.
  Thus, there is a node $w_i$ of $\pi$ with $S_i\in\ssi{w_i}$.  For a
  node $w$ closer to the root node than $w_i$, $S_i\in\ssb{w}$.  For a
  node $w'$ closer to the leaf node than $w_i$ along $\pi$,
  $S_i\not\in\sss{w'}$.  This also holds for $S_j$, so there is a node
  $w_j$ of $\pi$ with $S_j\in\ssi{w_j}$.  Without loss of generality,
  we assume that $w_i$ lies between the root node and $w_j$ (including
  them) along $\pi$.  Then we have $S_i\in\ssi{w_i}$ and
  $S_j\in\sss{w_i}$. Since $w_i$ is in $\pi$, $\slab{w_i}$ contains
  the left side of $Q$.  Therefore, $w_i$ is a canonical node of
  $(i,j,Q)$ in $\pi$. 
\end{proof}

We need the following lemma to bound the total number of canonical
nodes for $Q$ over all pairs of rectangles of $\iset$ by $O(k(Q))$.
Notice that the following lemma holds for a pair of rectangles of any
configuration from \confa to \confe.
\begin{lemma}\label{lem:canonical}
  For any rectangle $Q$ and any pair $(S_i,S_j)$ of rectangles of
  $\iset$ with $\itx{i}{j}\cap Q\neq\emptyset$, there is at most one
  canonical node of $(i,j,Q)$ in $T$.
\end{lemma}
\begin{proof}
  Let $v$ be a canonical node of $(i,j,Q)$ in $T$. Since the left side
  of $Q$ is contained in $\slab{v}$, the node $v$ is in the path $\pi$
  of $T$ from the root node to the leaf node $u$ such that the left
  side of $Q$ is contained in $\slab{u}$.  By the construction of the
  segment tree, there is at most one node $w_i$ on $\pi$ with
  $S_i\in\ssi{w_i}$, and there is at most one node $w_j$ on $\pi$ with
  $S_j\in\ssi{w_j}$.  Therefore, no node of $T$ other than $w_i$ and
  $w_j$ is a canonical node of $(i,j,Q)$.
	
  Without loss of generality, we assume that $v=w_i$. Then we have
  $S_j\in\sss{w_i}$ by the definition of the canonical node.
  If $S_j\in\ssi{w_i}$, we have
    $w_i=w_j$ and $w_i$ is the unique canonical node.
    If $S_j\in\ssb{w_i}$, $w_j$
    is not a canonical node of $(i,j,Q)$ because
  $w_i$ lies between the root node and $w_j$ (including
  the root node) along $\pi$ and $S_i\notin \sss{w_j}$.
Therefore, there is at most  one canonical node of $(i,j,Q)$.
\end{proof}

\begin{corollary}\label{cor:size-canonical}
  The total number of canonical nodes for a query rectangle $Q$ is
  $O(k(Q))$.
\end{corollary}
Our general strategy is the following. Given a query rectangle $Q$, we
find a set of nodes of the segment tree $T$ that contains the
canonical node of $(i,j,Q)$ for every \confe-pair $(S_i,S_j)$
not belonging to any other configuration such that the horizontal sides of 
	$\itx{i}{j}$ intersect the vertical sides of $Q$
in $O(\log n+k(Q))$ time. The
size of this set is $O(k(Q))$.  For each such node $v$, we find all
\confe-pairs $(S_i,S_j)$ such that $v$ is a canonical node of
$(i,j,Q)$ in time linear in the number of the output.

\subsubsection{Finding All Canonical Nodes for \confe-pairs}
\label{sec:T1}
In this subsection, we present data structures and their query
algorithms to find a set of canonical nodes of $(i,j,Q)$ with
$\itx{i}{j}\cap Q\neq\emptyset$ for a query rectangle $Q$.  This set
contains the canonical node of $(i,j,Q)$ for every \confe-pair
$(S_i,S_j)$ not belonging to any other configuration.  We show how to do this
for the \confe-pairs such that the horizontal sides of $\itx{i}{j}$
intersect the vertical sides of $Q$. 

\paragraph{Data Structures.}  For each node $v$ of $T$ and each
rectangle $S$ of $\sss{v}$, 
we define the \emph{trimmed rectangle} for
$(S,v)$ as the smallest rectangle containing $S_v\cap U(v)$, where
$S_v=S\cap H(v)$ and $U(v)=\bigcup_{S'\in\ssi{v}}S'$.
See Figure~\ref{fig:trimmed} for an illustration.
\begin{figure}
  \begin{center}
    \includegraphics[width=0.7\textwidth]{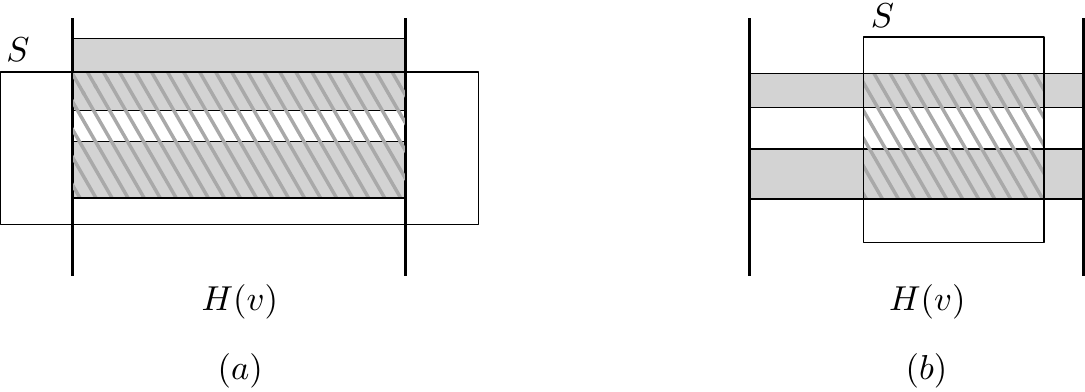}
    \caption{\label{fig:trimmed} A node $v$ and a rectangle $S\in\sss{v}$.
      The gray regions represent $H(v)\cap \bigcup_{S'\in\ssi{v}}S'$.
      (a) The trimmed rectangle (dashed region) for $(S,v)$ with $S\in\ssi{v}$.
      (b) The trimmed rectangle (dashed region) for $(S,v)$ with $S\in\ssb{v}$.}
  \end{center}
\end{figure}
Let $\mathcal{L}$ be the set of the horizontal sides of all trimmed
rectangles for all nodes of $T$.  Note that
$|\mathcal{L}|=O(n\log n)$. 
To compute $\mathcal{L}$ efficiently, we sort the rectangles of $\iset$ in
decreasing order with respect to their top sides in $O(n\log n)$ time. This allows us to sort
all rectangles of $\mathcal{S}(v)$ for each node $v$ of $T$ in the same depth in $O(n)$ time in total.
Therefore, we can sort the rectangles of $\mathcal{S}(v)$ in
decreasing order with respect to their top sides for every node $v$ of $T$ in $O(n\log n)$ time in total.
Similarly, we sort the rectangles of $\mathcal{S}_C(v)$ for every node $v$ of $T$ in $O(n\log n)$ time. 
The trimmed rectangle for $(S,v)$ is $S\cap H(v)$ for a rectangle $S$ of $\ssi{v}$.
For a rectangle $S$ of $\ssb{v}$, the top side of the trimmed rectangle for $(S,v)$ is the highest
top side of the rectangles of $\ssi{v}$ lying below the top side of $S$ if the top side of $S$ is not contained in any rectangle of $\ssi{v}$. Otherwise, the top side of the trimmed rectangle is the top side of $S$. 
Therefore, the top sides of the trimmed rectangles for $(S,v)$ can be computed in $O(|\mathcal{S}(v)|)$ time for a node $v$ of $T$
and all rectangles $S\in\mathcal{S}(v)$.
Thus we can compute $\mathcal{L}$ in time linear in its size, which is $O(n\log n)$.

We construct the hive-graph on $\mathcal{L}$, which allows us to
report all horizontal sides of $\mathcal{L}$ intersecting a query
vertical segment $\ell$ in sorted order along $\ell$ in $O(\log n+K)$
time, where $K$ is the size of output~\cite{chazelle1986Enclusure}.
Since the size of $\mathcal{L}$ is $O(n\log n)$, the hive-graph has
$O(n\log n)$ size. We make each segment of $\mathcal{L}$ to point to
the rectangle of $\iset$ from which the segment comes.

\paragraph{Query Algorithm.}  Given a query rectangle $Q$, our
query algorithm finds all sides of $\mathcal{L}$ intersecting the left
side of $Q$ using the hive-graph on $\mathcal{L}$. Then for each such
side, our query algorithm marks the node of $T$ pointed by the side as
a canonical node in $O(\log n+k)$ time due to the following lemmas.

\begin{lemma}\label{lem:witness}
  The query algorithm finds the canonical node of $(i,j,Q)$ for every
  \confe-pair $(S_i,S_j)$ not belonging to any other configuration
  such that the horizontal sides of 
  	$\itx{i}{j}$ intersect the vertical sides of $Q$.
\end{lemma}
\begin{proof}
  Consider a \confe-pair $(S_i,S_j)$ of $Q$ not belonging to any other
  configurations such that the horizontal sides of $\itx{i}{j}$
  intersect the vertical sides of $Q$. There is a unique canonical
  node $v$ of $(i,j,Q)$ by Lemma~3 and Lemma~4.  Let
  $S_i'$ and $S_j'$ be the trimmed rectangles for $(S_i,v)$ and
  $(S_j,v)$, respectively.

  We claim that a horizontal side of $S_j$ is intersected by the left
  side of $Q$.  Since $(S_i,S_j)$ belongs to \confe, the left side of
  $S_j$ lies to the left of $Q$, and the right side of $S_j$ lies to
  the right of $Q$. There are only two cases to consider: a horizontal
  side of $S_j$ is intersected by the left side of $Q$, or $S_j$
  contains $Q$. For the second case, $(S_i,S_j)$ belongs to \confa.
  This contradicts the assumption that $(S_i,S_j)$ does not belong to
  any configuration other than \confe. Thus the only possible case is
  the first one, and the claim holds.
  	
  Now we claim that a horizontal side of $S_j'$ is intersected by the
  left side of $Q$, and thus the query algorithm finds $v$ as the
  canonical node of $(i,j,Q)$.  Without loss of generality, we assume
  that the top side of $S_j$ is intersected by the left side of
  $Q$. The top side of $S_j'$ lies in between the top side of $S_j$ and the top side of $\itx{i}{j}$. Since the top side of $S_j$ and the top side of $\itx{i}{j}$ intersects the left side of $Q$, the claim holds.
\end{proof}

\begin{lemma}
  The number of the sides of $\mathcal{L}$ intersecting the left side
  of $Q$ is $O(k(Q))$.
\end{lemma}
\begin{proof}
  We use a charging scheme as follows.  We charge each horizontal side
  $\ell$ of $\mathcal{L}$ intersecting the left side $\ell_Q$ of $Q$
  to a pair $(S_i,S_j)\in\oset(Q)$ with $S_i\in\ssi{v}$ and
  $S_j\in\sss{v}$ such that both $S_i$ and $S_j$ contain the
  intersection point of $\ell$ and $\ell_Q$. If there are more than
  one such pair, we charge $\ell$ to an arbitrary one.
	
  We claim that there exists such a pair for every horizontal side of
  $\mathcal{L}$ intersecting the left side of $Q$.  Consider a
  horizontal side $\ell$ of $\mathcal{L}$.  Let $S_j$ be the rectangle
  of $\iset$ defining $\ell$. In other words, let $S_j$ be a rectangle
  of $\iset$ such that the trimmed rectangle $S_j'$ for $(S_j,v)$ has
  $\ell$ as its horizontal side for some node of $T$.  By the
  definition of the trimmed rectangle, a horizontal side of $S_j'$ is
  contained in some rectangle of $\ssi{v}$, say $S_i$. Thus, the
  intersection of the horizontal side of $S_j'$ and the left side of
  $Q$ is contained in $S_i$.  This means that $(S_i,S_j)\in\oset(Q)$.
	
  Now we claim that each pair $(S_i,S_j)\in\oset(Q)$ is charged at
  most once in this way.  In each node $v$, a pair $(S_i,S_j)$ is
  charged at most once. Moreover, $(S_i,S_j)$ is charged only in the
  canonical node of $(i,j,Q)$, which is unique by Lemma~4.  Therefore,
  $(S_i,S_j)$ is charged at most once, and the lemma holds.
\end{proof}

\begin{lemma}\label{lem:find}
  Given a query rectangle $Q$, we can find a set of at most $k$ nodes of $T$
  containing all canonical nodes for \confe-pairs not belonging to any
  other configuration in $O(\log n+k)$ time. 
\end{lemma}

\subsubsection{Handling Each Canonical Node to Find All \confe-pairs}
Let $\mathcal{V}_Q$ be the set of all nodes we found in
Section~\ref{sec:T1}.  For each node $v\in\mathcal{V}_Q$, we show how to
find all \confe-pairs $(S_i,S_j)$ not belonging to any other configuration
such that $v$ is a canonical node
of $(i,j,Q)$.  Here, we consider only the case that $S_i\in\ssi{v}$
and $S_j\in\sss{v}$.  The other case that $S_j\in\ssi{v}$ and
$S_i\in\sss{v}$ can be handled analogously. 
Moreover, we consider only the \confe-pairs such that 
the horizontal sides of $\itx{i}{j}$ intersect the vertical sides of $Q$.
The other case can be handled analogously.

For each node $v$, we
spend $O(1+k(v))$ time, where $k(v)$ is the number of the pairs
$(S_i,S_j)$ with $\itx{i}{j}\cap Q\neq\emptyset$ such that $v$ is
a canonical node of $(i,j,Q)$. Note that the sum of $k(v)$ for every node $v$ of $\mathcal{V}_Q$
is $O(k)$ by Lemma~\ref{lem:canonical}. Once we do this for every node in
$\mathcal{V}_Q$, we can obtain all \confe-pairs for the canonical nodes
of $(i,j,Q)$ not belonging to any other configuration in $O(k)$ time,
excluding the time for computing all such
canonical nodes.

\paragraph{Overall Strategy.}  Let $\mathcal{S}_Q$ be
the set of all rectangles
$S_j\in\sss{v}$ for each node $v\in\mathcal{V}_Q$ such that a horizontal
side of the trimmed rectangle for $(S_j,v)$ intersects the left side
of $Q$. We obtain $\mathcal{S}_Q$ while computing the
set $\mathcal{V}_Q$ in Section~\ref{sec:T1}.
Consider a \confe-pair $(S_i,S_j)$ not belonging to any other configuration
such that the horizontal sides of $\itx{i}{j}$ intersect the vertical side of $Q$. 
	The proof of Lemma~\ref{lem:witness} shows that a horizontal side of the trimmed rectangle for $(S_j,v)$ is intersected by the left side of $Q$, where $v$ is the canonical node of $(i,j,Q)$.
This means that $S_j$ is in $\mathcal{S}_Q$.
Since we already have $S_j$, the remaining task is to find $S_i$.

Given a node $v\in\mathcal{V}_Q$ and a rectangle $S_j$ in $\mathcal{S}_Q\cap \ssi{v}$,
we are to compute all rectangles $S_i\in\sss{v}$ 
with $\itx{i}{j}\cap Q\neq\emptyset$.
For a rectangle $S_i\in\sss{v}$ with $\itx{i}{j}\cap Q\neq\emptyset$, we observe that
$\mathsf{y}(S_j)$, $\mathsf{y}(S_i)$ and $\mathsf{y}(Q)$ 
contain a common point, where $\mathsf{y}(A)$ is the
the orthogonal projection of a set $A\subseteq \mathbb{R}^2$ onto the $y$-axis.
There are two cases: $\proj{y}{S_j}\cap\proj{y}{Q}$ contains an endpoint 
of $\proj{y}{S_i}$, or $\proj{y}{S_j}\cap\proj{y}{Q}$ is contained in $\proj{y}{S_i}$.

\paragraph{Data Structures and Preprocessing.}
We maintain two data structures, one for finding the rectangles of the
first case and the other for finding the rectangles of the second case.
The first one is organized as follows. For each node $v$ of $T$, we maintain two
sorted lists of the rectangles of $\ssi{v}$, one with respect to their
top sides and the other with respect to their bottom sides.  We make
each rectangle $S$ of $\ssb{v}$ to point to the rectangle of
$\ssi{v}$ with highest bottom side (and highest top side) lying below
the top side (and bottom side) of the trimmed rectangle for $(S,v)$.  Similarly, we make $S$ to
point to the rectangle of $\ssi{v}$ with lowest top side (and lowest
bottom side) lying above the bottom side (and top side) of the trimmed rectangle for $(S,v)$.

For the second one, we use a \emph{partially persistent data
  structure} of a linked list.  Once we update a linked
list and destroy the old versions, we cannot search any element on an
old version any longer.  But a partially persistent data structure
allows us to access any version at any time by keeping the changes on
the linked list.  Driscoll et al.~\cite{persistent} presented a
general method to make a data structure based on pointers partially
persistent.  Using their method, we can construct a partially
persistent data structure of a linked list.

In our case, the linked list has rectangles of $\ssi{v}$ as its
elements.  We consider a $y$-coordinate as a time stamp.  A rectangle
$S\in\ssi{v}$ is appended to the linked list at time $t_1$ and
is deleted from the linked list at time $t_2$, where $t_1$ is the
$y$-coordinate of the top side of $S$ and $t_2$ is the
$y$-coordinate of the bottom side of $S$.  Each insertion and
deletion can be done in constant time, which is subsumed in the
total preprocessing time.  For each horizontal side of $S_c\in\ssi{v}$, we need an
extra pointer that points to the first element of the persistent data
structure at time $t$, where $t$ is the $y$-coordinate of the
horizontal side.  The size of the partially persistent data structure
is linear in the size of $\sss{v}$. 
Due to the partially
persistent data structure and the pointers associated with
the horizontal sides of the rectangles of $\sss{v}$, we can find all rectangles of $\sss{v}$
containing a horizontal side of $S_c\in\ssi{v}$ in time linear in the size of the output.

\paragraph{Query Algorithm.}
Given a node $v\in\mathcal{V}_Q$ and a rectangle $S_j$ in $\mathcal{S}_Q\cap \ssi{v}$,
we are to compute all rectangle $S_i\in\sss{v}$ 
such that $\mathsf{y}(S_j)$, $\mathsf{y}(S_i)$ and $\mathsf{y}(Q)$ contain a common point.
Recall that there are two cases: $\proj{y}{S_j}\cap\proj{y}{Q}$ contains an endpoint 
of $\proj{y}{S_i}$, or $\proj{y}{S_j}\cap\proj{y}{Q}$ is contained in $\proj{y}{S_i}$.
A horizontal side of the 
trimmed rectangle for $(S_j,v)$ is intersected by the left side of $Q$ by the definition of $\mathcal{S}_Q$. Thus at least one endpoint of $\proj{y}{S_j}$ is contained in
$\proj{y}{Q}$.
We assume that the endpoint of $\proj{y}{S_j}$ with smaller $y$-coordinate is contained in $\proj{y}{Q}$. In other words, the bottom side of $S_j$ intersects $Q$.
The other case can be handled analogously.

To find the rectangles $S_i$ belonging to the first case, we do the followings.
We search the sorted list of the rectangles of
$\ssi{v}$ with respect to their top sides starting from the rectangle
of $\ssi{v}$ with lowest top side lying above the bottom side of $S_j$. 
Note that we can obtain the starting point using the pointer
that the bottom side of $S_j$ has.  We stop searching the sorted list when we reach the top side
of $S_j$ or the top side of $Q$. 
In this way, we can find all rectangles $S_i$ of $\ssi{v}$ belonging to the first case 
in $O(1+K)$ time, where $K$ is the number of such rectangles.

To find the rectangles $S_i$ belonging to the second case, we do the followings.
A rectangle $S_i$ belonging to the second case intersects the bottom side $\ell$ of $S_j$. 
We search the partially persistent data structure at time $t$, where
$t$ is the $y$-coordinate of $\ell$. Specifically, 
starting from the pointer that $\ell$ points to, we
traverse the linked list at time $t$. All rectangles we encounter
are the rectangles containing $\ell$.  This
takes $O(1+K')$ time, where $K'$ is the number of such rectangles.

In total, we spend $O(1+k(v))$ time for each node $v\in\mathcal{V}_Q$, where $k(v)$ is the number
of the pairs $(S_i,S_j)$ of $\oset(Q)$ such that the canonical node of $(i,j,Q)$ is $v$.
Note that $k(v)$ is at least one for every node $v\in\mathcal{V}_Q$ by the construction of $\mathcal{V}_Q$. 
Once we do this for every node
in $\mathcal{V}_Q$, we can obtain $\oset(Q)$ in $O(1+k(Q))$ time in
total.
\begin{lemma}
  Given a query rectangle $Q$, we can find all \confe-pairs in
  $O(\log n+k(Q))$ time.
\end{lemma}

Therefore, we have the following theorem.
\begin{theorem}\label{thm:planar}
  We can construct a data structure of size $O(n\log n)$ on a set
  $\iset$ of $n$ axis-parallel rectangles so that for a query
  axis-parallel rectangle $Q$, the pairs $(S_i,S_j)$ of $\iset$ with
  $S_i\cap S_j\cap Q\neq\emptyset$ can be reported in $O(\log n+k)$ time,
  where $k$ is the size of the output.
\end{theorem}

\section{Higher Dimensional Case}\label{sec:dim}
In this section, we consider a set $\iset=\{S_1,\ldots,S_n\}$ of $n$
axis-parallel boxes in $\mathbb{R}^d$ for a constant $d>2$. 
Let $\delta\in\mathbb{R}$ be any constant with $1/d\leq\delta<1$.
We present a data structure that
supports $O(n^{1-\delta}\log^{d-1} n + k \polylog n)$ query time.
The size of the data structure is $O(n^{\delta d -2\delta+1}\log n)$.
There has been no known result for this problem in higher dimensions,
except that for $d=3$, the best known data structure has size of
$O(n\sqrt {n}\log n)$ and supports $O(\sqrt{n}\log^2n+k\log^2n)$ query
time~\cite{deBerg2015-journal}.

\subsection{Data Structure}
We denote the $t$th axis of $\mathbb{R}^d$ by the
\emph{$\axis{t}$-axis} for $1\leq t\leq d$.
The $\axis{t}$-\emph{projection} of a point set $A\subseteq\mathbb{R}^d$
is defined as the orthogonal projection of $A$ onto the $\axis{t}$-axis.
A box is given in the form
$\{(\axis{1},\axis{2},\ldots,\axis{d}) \mid a_t\leq \axis{t} \leq b_t , 1\leq t\leq d\}$
and has $2d$ facets. We call a facet of the box orthogonal to the $\axis{t}$-axis an
\emph{$\axis{t}$-facet} of the box for any $1\leq t\leq d$.
Our data
structure consists of the following substructures. We denote the
combination of them by $\BoxPairInt{d}$.

\paragraph{$n^{\delta}$-Clustered Grid Cells.}
For each index $1\leq t\leq d$, we construct $O(n^{\delta})$ intervals
on the $\axis{t}$-axis.  Consider the $\axis{t}$-projection 
of the $\axis{t}$-facets of the boxes of $\iset$,
which forms $2n$ points on $\axis{t}$. We choose every
$\lfloor n^{1-\delta}\rfloor$th points in the projection. 
Then we have
$O(n^\delta)$ points in the projection that define $O(n^{\delta})$ intervals
containing no chosen points in its interior. Let $\mathcal{I}_t$ be the
set of such intervals.  A \emph{grid cell} is a $d$-tuple
$(I_1,\ldots,I_d)$ of intervals $I_t\in \mathcal{I}_t$ for $1\leq t\leq d$.
Note that there are $O(n^{\delta d})$ grid cells.
For a box $B$ in $\mathbb{R}^d$, not necessarily in $\iset$, we call
the grid cell containing the corner of $B$ with
minimum $\axis{t}$-coordinates for all $1\leq t\leq d$
the \emph{canonical grid cell} of a box $B$. Every box in $\mathbb{R}^d$ has a unique
canonical grid cell.

\paragraph{Grid Containment Data Structure: {\mdseries\GridCont}. }
We mark a grid cell if it is the canonical grid cell of $\itx{i}{j}$
for a pair $(S_i,S_j)$.  We construct the \emph{grid containment data
  structure} on the marked grid cells, denoted by \GridCont, that
allows us to find all marked grid cells contained in a query
axis-parallel box.  To do this, we compute the largest box $Q'$
contained in $Q$ and aligned to the grid in $O(d\log n)$ time.
Specifically, for each $1\leq t\leq d$, we compute the union of all
intervals of $\mathcal{I}_t$ on the $\axis{t}$-axis contained in the
$\axis{t}$-projection of $Q$ in $O(\log n)$ time by applying binary
search on the intervals of $\mathcal{I}_t$. Then $Q'$ is the box whose
$\axis{t}$-projection is the union on the $\axis{t}$-interval for
every $1\leq t\leq d$. Then it suffices to find every marked grid cell
having its corner contained in $Q'$.  We construct a data structure of
size $O(n(\log n/\log\log n)^{d-1})$ on the corners of all marked grid
cells so that for any query axis-parallel box, the corners contained
in the query box can be reported in $O(\log^{d-1}n +K)$ time, where
$K$ is the size of the output~\cite{afshani2012higher}.  Since each
marked grid cell is reported exactly $2^d$ times and there are
$O(k(Q))$ marked grid cells contained in $Q$, we can find all marked
grid cells contained in $Q$ in $O(\log^{d-1} n+k(Q))$ time.

\paragraph{Box Intersection Data Structure: {\mdseries\BoxInt}.}
We construct a data structure, denoted by \BoxInt, of size $O(n\log^{d-2} n)$ 
that allows us to report the boxes of $\iset$
intersecting a query axis-parallel box in $O(\log^{d-1} n+K)$ time as follows,
where $K$ is the size of the output.

A box $S$ of $\iset$ intersects any query axis-parallel box $Q$ in
$\mathbb{R}^d$ if and only if one of the following holds: $S$ contains
a corner of $Q$, a corner of $S$ is contained in $Q$, or a facet of
$S$ intersects $Q$.  For the first case, we maintain the data
structure given by Chazelle~\cite{chazelle1986Enclusure} of size
$O(n\log^{d-2}n)$ that allows us to find all boxes of $\iset$
containing a query point (a corner of $Q$) in $O(\log^{d-1} n+K)$
time.  For the second case, we use the data structure given by Afshani
et al.~\cite{afshani2012higher} of size
$O(n(\log n/\log\log n)^{d-1})$ that allows us to find all corners of
$B$ contained in a query box in $\mathbb{R}^d$ in
$O(\log n(\log n/\log\log n)^{d-4+1/(d-1)}+K)$ time.
  
For the third case, we construct a data structure recursively using
the data structure described in Section~2 as a base structure.  An
$\axis{t}$-facet of $S$ intersects $Q$ if and only if the
$\axis{t}$-projection of the facet is contained in the
$\axis{t}$-projection of $Q$ and the projection of the facet onto a
hyperplane orthogonal to the $\axis{t}$-axis intersects the projection
of $Q$ onto the hyperplane.  To use this property, we compute the
$\axis{t}$-projection of every $\axis{t}$-facet of the boxes of
$\mathcal{S}$ and denote the set of them by $\mathcal{P}_t$ for each
$1\leq t\leq d$.  Since the $\axis{t}$-axis is orthogonal to
$\axis{t}$-facets, each projection is a point on the $\axis{t}$-axis.
We construct a one-dimensional range tree $\mathcal{T}_t$ (a balanced
binary search tree) on $\mathcal{P}_t$ for each $1\leq t\leq
\ell$. Each node $v$ of $\mathcal{T}_t$ is associated with a set
$\mathcal{S}(v)$ of boxes $S$ of $\mathcal{S}$ such that the
$\axis{t}$-projection of an $\axis{t}$-facet of $S$ is contained in the
interval of the $\axis{t}$-axis corresponding to the node. We
recursively construct the $(d-1)$-dimensional data structure on the
projections of the boxes of $\mathcal{S}(v)$ onto a hyperplane
orthogonal to the $\axis{t}$-axis. Let $\mathcal{V}$ denote the set of
the nodes in the range trees $\mathcal{T}_t$ for all indices
$1\leq t\leq d$.
Assume that given a set of $N$ axis-parallel boxes in
$\mathbb{R}^{d-1}$ for some $3\leq \ell < d$, we can construct a data
structure of size $S(N,d-1)$ that allows us to find all input boxes
intersecting a query $(d-1)$-dimensional axis-parallel box in
$T(N,d-1)$ time.  We have
\[ S(n, d) =
  \begin{cases}
    \sum_{v\in\mathcal{V}}  S(|\mathcal{S}(v)|, d-1)       & \quad \text{if } d > 3\\
    O(n\log n) & \quad \text{if } d=3.\\
  \end{cases}
\]
Moreover, since for any box of $\iset$ and any index $1\leq t\leq d$,
there are $O(\log n)$ nodes $v$ of $\mathcal{T}_t$ such that the box
is contained in $\mathcal{S}(v)$, we have
\[
  \sum_{v\in\mathcal{V}} |\mathcal{S}(v)| = O(dn\log n)= O(n\log n).
\]
Thus, the size of the data structure for the $d$-dimensional space is
$O(n\log^{d-2} n)$.
  	
Now we show that we can find all boxes of $\iset$ whose facets
intersect $Q$ using this data structure constructed on $\mathcal{S}$.
For each $1\leq t\leq d$, we find all boxes of $\iset$ whose
$\axis{t}$-facets intersect $Q$.  To do this, we consider the range
tree $\mathcal{T}_t$ and find $O(\log n)$ nodes $v$ such that the
interval corresponding to $v$ is contained in the
$\axis{t}$-projection of $Q$, but the interval corresponding to the
parent of $v$ is not contained in the $\axis{t}$-projection of $Q$.
Let $\Pi_t$ denote the set of such nodes for an index $t$ and $\Pi$
denote $\bigcup_{1\leq t\leq d}\Pi_t$.
	
For each node $v$ in $\Pi_t$, a box $S$ of $\mathcal{S}(v)$ has an
$\axis{t}$-facet intersecting $Q$ if and only if the projection of $S$
onto a hyperplane $h$ orthogonal to the $\axis{t}$-axis intersects the
projection of $Q$ onto $h$.  Thus we can find all boxes of $S$ with
$\axis{t}$-facets intersecting $Q$ using the $(d-1)$-dimensional data
structure associated with each such node.  We have
\[ T(n, d) =
  \begin{cases}
    \sum_{v\in\Pi}  (T(|\mathcal{S}(v)|, d-1)+\log n)       & \quad \text{if } d > 3\\
    O(\log n+\sum_{v\in\Pi} k_v) & \quad \text{if } d=3\\
  \end{cases}
\]
where $k_v$ is the number of boxes of $\mathcal{S}(v)$ whose
$x_t$-facets intersect $Q$ for an index $1\leq t\leq d$ and a node $v$
of $\mathcal{T}_t$.  Since for any box of $\iset$, there are at most
one node $v$ in $\Pi_t$ such that the box is contained in
$\mathcal{S}(v)$ for each index $1\leq t\leq d$, we have
\[
  \sum_{v\in\Pi} k_v = O(dK)=O(K) \text{ and } |\Pi|=O(d\log n)=O(\log
  n).
\]
	
Thus, the query algorithm for the $d$-dimensional case takes
$O(\log^{d-1}n+K)$ time.

\paragraph{Pair Finding Data Structure: {\mdseries\PairFind}.}
Recall that we mark the canonical grid cell of $\itx{i}{j}$ for each
pair $(S_i,S_j)$ of boxes of $\iset$.  However, we do not store the
pair to each canonical grid cell explicitly.  Otherwise, the size of
the data structure becomes $\Theta(n^2)$.  Instead, we present an
efficient way together with a data structure, denoted by \PairFind, to
find all pairs $(S_i,S_j)$ of $\iset$ such that the canonical grid
cell of $\itx{i}{j}$ is a given grid cell.  Specifically, we present a
data structure of size $O(n\log^{d-2} n)$ supporting
$O(\log^{d-1} n +K)$ query time, where $K$ is the size of the output.

Let $\Box$ be a given grid cell.  Recall that the canonical grid cell
of $\itx{i}{j}$ is the grid cell containing the corner $c$ of
$\itx{i}{j}$ with minimum $\axis{t}$-coordinates for all
$1\leq t\leq d$.  Let $f_t$ be the $\axis{t}$-facet of $\itx{i}{j}$
incident to $c$ for an index $1\leq t\leq d$.  Note that $f_t$ comes
from $S_i$ or $S_j$, that is, $f_t$ is contained in an
$\axis{t}$-facet of $S_i$ or $S_j$.
	
Let $F_i$ be any subset of $\{1,\ldots,d\}$, and
$F_j=\{1,\ldots,d\}\setminus F_i$.  There are $2^d$ possible pairs
$(F_i,F_j)$ of the sets.  We handle each case one by one, and find all
pairs $(S_i,S_j)$ of $\iset$ such that $f_t$ comes from $S_i$ for
every index $t\in F_i$ and $f_{t'}$ comes from $S_j$ for every index
$t'\in F_j$.
Note that $S_i$ has two $\axis{t}$-facets. By the definition of the
canonical grid cell, $f_t$ comes from the $\axis{t}$-facet of $S_i$
with smaller $\axis{t}$-coordinate.

Given a pair $(F_i,F_j)$, we first find all boxes of $\iset$ whose
$\axis{t}$-facets with smaller $\axis{t}$-coordinate intersect $\Box$
for all $t\in F_i$. The $\axis{t}$-facet of a box $S$ of $\iset$ with
smaller $\axis{t}$-coordinate intersects $\Box$ for all $t\in F_i$ if
and only if the common intersection of all $\axis{t}$-facets
intersects $\Box$. Note that the common intersection is a $(d-t)$-face
of $S$.  To find all such boxes, in the preprocessing phase, we map
each box $S$ of $\iset$ to the common intersection of the
$\axis{t}$-facets of $S$ with smaller coordinates for all $t\in F_t$.
Then the problem reduces to the problem of finding all $(d-t)$-faces
of boxes of $\iset$ 
intersecting a query box. This takes $O(\log^{d-1} n+k)$ time using
$O(n\log^{d-2} n)$ space by constructing $\BoxInt$ on all the
$(d-t)$-faces. Note that a $(d-t)$-face of a box of $\iset$ is also an
axis-parallel box in $\mathbb{R}^d$.  Also, we can check whether there
is such a box in $O(\log^{d-1} n)$ time.  Let $\iset_i$ be the set of
all such boxes.  Similarly, we check whether there is a box of $\iset$
whose $\axis{t'}$-facet with smaller $\axis{t'}$-coordinate intersect
$\Box$ for all $t'\in F_j$. Let $\iset_j$ be the set of such boxes.
If both $\iset_i$ and $\iset_j$ are nonempty, we find them explicitly
and report them as pairs $(S_i,S_j)$ with $S_i\in\iset_i$ and
$S_j\in\iset_j$ such that the canonical grid cell of $\itx{i}{j}$ is a
given grid cell in $O(\log^{d-1} n +K)$ time, where $K$ is the size of
the output.

\paragraph{Facet Intersecting Data Structure: {\mdseries$\BoxPairInt{d-1}$}.}
For each interval $I$ of $\mathcal{I}_t$ for an index $1\leq t\leq d$,
we construct a $(d-1)$-dimensional data structure for our problem.
Consider the boxes of $\iset$ whose $\axis{t}$-projections contain
$I$. We compute the projections of such boxes onto a hyperplane
orthogonal to the $\axis{t}$-axis.  These projections are boxes in
$\mathbb{R}^{d-1}$. Then we construct a $(d-1)$-dimensional data
structure $\BoxPairInt{d-1}$ on these boxes.  For $d=2$, we use the
data structure of size $O(n\log n)$ described in Section~2.

\begin{lemma}
  The size of $\BoxPairInt{d}$ is $O(n^{\delta d-2\delta+1}\log n)$.
\end{lemma}
\begin{proof}
  The size of \GridCont is $O(n(\log n/\log\log n)^{d-1})$, the size
  of \BoxInt is $O(n\log^{d-2} n)$, and the size of \PairFind is
  $O(n\log^{d-2} n)$. Also, we construct $\BoxPairInt{d-1}$ on each
  interval of $\mathcal{I}_t$ for each index $1\leq t\leq
  d$. Therefore, we have the following recurrence.  Let $S(n,d)$ be
  the size of $\BoxPairInt{P}$ constructed on $n$ axis-parallel boxes.
  \[
    S(n,d)= O(n(\log n/\log\log n)^{d-1}) + n^{\delta}d \cdot
    S(n,d-1).
  \]
  Since $d$ is a constant and $S(n,2)=O(n\log n)$, we have
  $S(n,d)= O(n^{\delta d-2\delta +1} \log n)$.
\end{proof}

\subsection{Query Algorithm}
Given a query of an axis-parallel box $Q$, we present an algorithm for
finding all pairs $(S_i,S_j)$ of boxes of $\iset$ with
$\itx{i}{j}\cap Q\neq\emptyset$. We observe that the canonical grid
cell of $\itx{i}{j}$ is contained in $Q$, or $\itx{i}{j}$ intersects a
grid cell intersecting the boundary of $Q$ for such a pair
$(S_i,S_j)$.  To see this, consider the union of the grid
  cells intersecting the interior of $Q$ but not intersecting the
  boundary of $Q$. The union is a box in $\mathbb{R}^d$ contained in
  $Q$.  If $\itx{i}{j}$ is contained in this union, the canonical grid
  cell of $\itx{i}{j}$ is also contained in this union and $Q$.  If
  $\itx{i}{j}$ is not contained in this union, $\itx{i}{j}$ intersects
  a grid cell intersecting the boundary of $Q$.

\paragraph{Case 1: The Canonical Grid Cell of $\itx{i}{j}$ is Contained in $Q$.}
To find every pair $(S_i,S_j)$ of boxes of $\iset$ such that the canonical grid cell of $\itx{i}{j}$ 
is contained in $Q$, 
we find all marked grid cells contained in $Q$ using $\GridCont$
in $O(\log^{2d-2} n+k(Q))$ time. Note that the size of the output
is at most $k(Q)$ since we consider the \emph{marked} grid cells only. 
For each such grid cell $\Box$, we find all pairs $(S_i,S_j)$ of boxes of $\iset$
such that the canonical grid cells of $\itx{i}{j}$ 
are $\Box$
in $O(\log^{d-1} n+ k(Q))$ time using \PairFind.
Therefore, it takes $O(k(Q)\log^{d-1} n +\log^{2d-2}n)$ time in total.

\paragraph{Case 2: $\itx{i}{j}$ Intersects a Grid Cell Intersecting the Boundary of $Q$.}
Consider the interval we constructed on the $\axis{t}$-axis
containing the $\axis{t}$-projection (point) of an $\axis{t}$-facet $f$ of $Q$
for an index $1\leq t\leq d$.  Let $H$ be the union of all grid cells whose
$\axis{t}$-projections are this interval. Note that $H$ is a slab orthogonal to the $\axis{t}$-axis.
  We show how to find all
pairs such that $\itx{i}{j}$ intersects $H$ and
$\itx{i}{j}\cap Q\neq\emptyset$. The other cases can be handled
analogously.

Consider a pair $(S_i,S_j)$ such that $\itx{i}{j}$ intersects
$H$. Either one of $S_i$ and $S_j$ has an $\axis{t}$-facet contained in $H$, or both $S_i$ and $S_j$  cross  
$H$. Moreover, there are $O(n^{1-\delta})$ boxes of $\iset$ having their $\axis{t}$-facets contained in $H$ by the construction of the grid cells. For each box
$S$ which has an $\axis{t}$-facet contained in $H$, we find all boxes $S'$
of $\iset$ intersected by $S\cap Q$ using \BoxInt in
$O(\log^{d-1} n+ K)$ time, where $K$ is the size of the output. We can
do this for all boxes belonging to the first type in $O(n^{1-\delta}\log^{d-1} n+ k(Q))$
time.

For the pairs $(S_i,S_j)$ such that $S_i$ and $S_j$ cross $H$, we use $\BoxPairInt{d-1}$ associated with $H$. For any two boxes $S_i$ and $S_j$ of $\iset$ crossing $H$, we have $\itx{i}{j}\cap Q\neq\emptyset$ if and only if $\proj{h}{S_i}\cap\proj{h}{S_j}\cap\proj{h}{Q}\neq\emptyset$,
where $\proj{h}{A}$ denotes the projection of a set $A\subseteq\mathbb{R}^d$ onto a hyperplane
orthogonal to the $\axis{t}$-axis. 
This means that the problem reduces to the $(d-1)$-dimensional problem.
We find all pairs $(S_i,S_j)$ of the boxes of $\iset$ crossing $H$ such that $\proj{h}{S_i}\cap\proj{h}{S_j}\cap\proj{h}{Q}\neq\emptyset$. 
Therefore, we find all pairs $(S_i,S_j)$ of $\iset$ such that $\itx{i}{j}$ intersects
a grid cell intersecting the boundary of $Q$.

\paragraph{Analysis of the Running Time.} Let $T(n,k,d)$ denote
the running time of our algorithm in $d$-dimensional
space with input size $n$ and output size $k$.
Then we have the following recurrence relation.
\[
T(n,k,d)  =  O(n^{1-\delta}\log^{d-1} n) +O(k' \log^{d-1} n) + \sum_{1\leq i\leq d} T(n,k_i,d-1),
\]
where the sum of $k'$ and all $k_i$'s is $O(k(Q))$.
By Theorem~10, we have $T(n,k,2)= O(\log n+ k)$. By solving the recurrence relation,
we have the following theorem.
\begin{theorem}\label{thm:high}
  We can construct data structures on a set $\iset$ of $n$
  axis-parallel boxes in $\mathbb{R}^d$ for a constant $d$ so that for
  a query axis-parallel box $Q$, the pairs $(S_i,S_j)$ of boxes of
  $\iset$ with $S_i\cap S_j\cap Q\neq\emptyset$ can be reported in
  $O(n^{1-\delta}\log^{d-1} n + k\log^{d-1} n)$ time, where $k$ is the
  size of the output. The size of the data structure is
  $O(n^{\delta d - 2\delta +1}\log n)$.
\end{theorem}

\bibliography{paper}{}

\end{document}